\def\year{2018}\relax
\documentclass[letterpaper]{article} 
\usepackage{aaai18}  
\usepackage{times}  
\usepackage{helvet}  
\usepackage{courier}  
\usepackage{url}  
\usepackage{graphicx}  
\usepackage{booktabs}
\usepackage{multirow}

\usepackage{microtype}

\usepackage{url}
\usepackage[usenames,dvipsnames]{color}
\usepackage{tikz}

\DeclareRobustCommand{\DE}[3]{#2}

\newcommand{\citet}[1]{\citeauthor{#1} (\citeyear{#1})}
\newcommand{\citep}[1]{(\citeauthor{#1}, \citeyear{#1})}
\newcommand{\citealt}[1]{\citeauthor{#1} \citeyear{#1}}

\usepackage{amsthm}
\usepackage{amssymb}

\newtheorem{theorem}{Theorem}
\newtheorem{theorem*}[theorem]{Theorem$^{\star}$}

\newtheorem{lemma*}[theorem]{Lemma$^{\star}$}
\newtheorem{proposition}[theorem]{Proposition}
\newtheorem{proposition*}[theorem]{Proposition$^{\star}$}
\newtheorem{observation}[theorem]{Observation}
\newtheorem{corollary}[theorem]{Corollary}
\newtheorem{corollary*}[theorem]{Corollary$^{\star}$}
\newtheorem{claim}[theorem]{Claim}

\usepackage{stmaryrd}
\usepackage{mathtools}

\newcommand{\SB}{\{\,}%
\newcommand{\SM}{\;{:}\;}%
\newcommand{\SE}{\,\}}%
\newcommand{\SBs}{\{}%
\newcommand{\SEs}{\}}%

\renewcommand{\P}{\text{\normalfont P}}
\newcommand{\NP}{\text{\normalfont NP}}
\newcommand{\PSPACE}{\text{\normalfont PSPACE}}
\newcommand{\EXPTIME}{\text{\normalfont EXPTIME}}

\newcommand{\FPT}{\text{\normalfont FPT}}
\newcommand{\XP}{\text{\normalfont XP}}
\newcommand{\co}{\text{\normalfont co-}}
\newcommand{\para}{\text{\normalfont para-}}
\newcommand{\W}[1]{\text{\normalfont W[#1]}}

\newcommand{\NN}{\mathbb{N}}

\newcommand{\AAA}{\mathcal{A}}

\newcommand{\BBB}{\mathcal{B}}
 
\newcommand{\EEE}{\mathcal{E}}
\newcommand{\LLL}{\mathcal{L}} 

\newcommand{\III}{\mathcal{I}}

 \newcommand{\TTT}{\mathcal{T}}

\newcommand{\Card}[1]{|#1|}

\newcommand{\mtext}[1]{\text{\normalfont #1}}

\newcommand{\SAT}{\ensuremath{\mtext{\sc SAT}}}
\newcommand{\TQBF}{\ensuremath{\mtext{\sc TQBF}}}

\newcommand{\Clique}{\mtext{\textsc{Clique}}}

\newcommand{\Var}[1]{\mtext{Var\ensuremath{(#1)}}}

\newcommand{\ALC}[0]{\ensuremath{\mathcal{ALC}}}
\newcommand{\ALE}[0]{\ensuremath{\mathcal{ALE}}}
\newcommand{\ALU}[0]{\ensuremath{\mathcal{ALU}}}
\newcommand{\AL}[0]{\ensuremath{\mathcal{AL}}}

\newcommand{\EL}[0]{\ensuremath{\mathcal{E\!L}}}
\newcommand{\ELI}[0]{\ensuremath{\mathcal{E\!LI}}}

\newenvironment{myquote}{\begin{center}
    \begin{minipage}{.98\linewidth}}{\end{minipage}\end{center}}

\begin{document}
%
\title{A Parameterized Complexity View on Description Logic Reasoning}
\author{Ronald de Haan\\
Institute for Logic, Language and Computation\\
University of Amsterdam\\
\url{me@ronalddehaan.eu}
}
\maketitle

\begin{abstract}
Description logics are knowledge representation languages that have been
designed to strike a balance between expressivity and
computational tractability.
Many different description logics have been developed,
and numerous computational problems for these logics have been studied
for their computational complexity.
However, essentially all complexity analyses of reasoning
problems for description logics use the one-dimensional framework
of classical complexity theory.
The multi-dimensional framework of parameterized complexity theory
is able to provide a much more detailed image of the complexity
of reasoning problems.

In this paper we argue that the framework of parameterized complexity
has a lot to offer for the complexity analysis of description logic reasoning
problems---when one takes a progressive and forward-looking view
on parameterized complexity tools.
We substantiate our argument by means of three case studies.
The first case study is about the problem of concept satisfiability
for the logic \ALC{} with respect to nearly acyclic TBoxes.
The second case study concerns concept satisfiability
for \ALC{} concepts parameterized by the number of occurrences
of union operators and the number of occurrences of
full existential quantification.
The third case study offers a critical look at data complexity results
from a parameterized complexity point of view.
These three case studies are representative for the wide range of
uses for parameterized complexity methods for description
logic problems.
\end{abstract}

\section{Introduction}


Description logics have been designed as
knowledge representation formalisms that have good computational
properties \cite{BaaderCalvaneseMcGuinnessNardiPatelSchneider03}.
Correspondingly, there has been a lot of research into the computational
complexity of reasoning problems for different description logics.
This research has, however, focused entirely on the framework of classical
complexity theory to study the computational complexity
(see, e.g.,~\citealt{BaaderCalvaneseMcGuinnessNardiPatelSchneider03};~%
\citealt{BaaderHorrocksSattler08}).

The more fine-grained and multi-dimensional
framework of parameterized complexity theory
has hardly been applied to study the complexity of reasoning problems
for description logics.
Only a few works used the framework of parameterized complexity
to study description logic problems
\cite{BienvenuKikotKontchakovPodolskiiRyzhikovZakharyaschev17,%
BienvenuKikotKontchakovRyzhikovZakharyaschev17,%
CeylanPenaloza14,KikotKontchakovZakharyaschev11,Motik12,%
SimancikMotikHorrocks14,SimancikMotikKroetzsch11}.
Moreover, these works all use the framework in a traditional way,
focusing purely on one commonly used notion of tractability
(namely that of fixed-parameter tractability).

Parameterized complexity is designed to address the downside of
classical complexity theory that it is largely ignorant of structural properties
of problem inputs that can potentially be exploited algorithmically.
It does so by distuinguishing a problem parameter~$k$, in addition to the
input size~$n$, and measuring running times in terms of both of these.
The parameter~$k$ can be used to measure various types of structure
that are present in the problem input.
Parameterized complexity theory has grown into a large
and thriving research community over the last few decades
(see, e.g.,~\citealt{BodlaenderDowneyFominMarx12};
\citealt{Downey12};
\citealt{DowneyFellows13}).
Most results and techniques in parameterized complexity
theory revolve around the notion of \emph{fixed-parameter tractability}---%
a relaxation of polynomial-time solvability
based on running times of the
form~$f(k) \cdot n^{O(1)}$, for some computable
function~$f$ (possibly exponential or worse).

Due to the fact that reasoning problems related to
description logics are typically of high complexity
(e.g., complete for classes like \PSPACE{} and \EXPTIME{}),
it is unsurprising that one would
need very restrictive parameters to obtain fixed-parameter tractability results
for such problems.
It has been proposed recently that the investigation of
problems that are of higher complexity
can also benefit from the parameterized complexity point of view
\cite{DeHaan16,DeHaanSzeider14b,DeHaanSzeider14,DeHaanSzeider16,%
DeHaanSzeider17}---using tools and methods
that overstep the traditional focus on fixed-parameter tractability as positive results.

In this paper, we show how the complexity study of description logic
problems can benefit from using the framework of parameterized complexity
and all the tools and methods that it offers.
We do so using three case studies:
(1)~parameterized results for
concept satisfiability for \ALC{} with respect to nearly acyclic TBoxes,
(2)~parameterized results for
concept satisfiability for fragments of \ALC{} that are close
to \ALE{}, \ALU{} and \AL{}, respectively, and
(3)~parameterized results
addressing the notion of data complexity
for instance checking and conjunctive query entailment for \ELI{}.
The complexity results that we obtain
are summarized in Tables~\ref{table:alc-concept-sat-tboxes},%
~\ref{table:alc-concept-sat} and~\ref{table:data-complexity}---%
at the end of the sections where we present the case studies.


%

\paragraph{Outline.}

We begin by giving an overview of the theory of parameterized complexity---%
including commonly used (and more traditional) concepts and tools,
as well as more progressive notions.
Then we present our three case studies in three separate sections,
before sketching directions for future research and concluding.
%

\section{Parameterized Complexity Theory}

We begin by introducing relevant concepts 
from the theory of parameterized complexity.
For more details, we refer to textbooks on the 
topic~\cite{DowneyFellows13,FlumGrohe06}.
We introduce both concepts that are used commonly in
parameterized complexity analyses in the literature
and less commonly used concepts, that play a role
in this paper.

\paragraph{FPT and XP.}
The core notion in parameterized complexity is that
of fixed-parameter tractability, which is a relaxation of the
traditional notion of polynomial-time solvability.
Fixed-parameter tractability is a property of parameterized problems.
A \emph{parameterized problem~$Q$} is a subset of~$\Sigma^{*} \times \NN$,
for some finite alphabet~$\Sigma$.
An instance of a parameterized problem is a pair~$(x,k)$
where~$x$ is the main part of the instance,
and~$k$ is the parameter.
Intuitively, the parameter captures some type of structure
of the instance that could
potentially be exploited algorithmically---%
the smaller the value of the parameter~$k$, the more structure
there is in the instance.
(When considering multiple parameters,
we take their sum as a single parameter.)
A parameterized problem is \emph{fixed-parameter tractable}
if instances~$(x,k)$ of the problem can be solved by a deterministic algorithm
that runs in time~$f(k)\Card{x}^{O(1)}$,
where~$f$ is a computable function of~$k$.
Algorithms running within such time bounds are called \emph{fpt-algorithms}.
\FPT{} denotes the class of all parameterized problems
that are fixed-parameter tractable.

Intuitively, the idea behind fixed-parameter tractability is that whenever the
parameter value~$k$ is small, the overall running time is reasonably small---%
assuming that the constant hidden behind~$O(1)$ is small.
In fact, for every fixed parameter value~$k$, the running time of an fpt-algorithm
is polynomial (where the order of the polynomial is constant).

A related parameterized complexity class is \XP{},
which consists of all parameterized problems for which instances~$(x,k)$
can be solved in time~$n^{f(k)}$, for some computable function~$f$.
Algorithms running within such time bounds are called \emph{xp-algorithms}.
That is, a parameterized problem~$Q$ is in \XP{} if there is an algorithm
that solves~$Q$ in polynomial time for each fixed value~$k$ of the parameter---%
where the order of the polynomial may grow with~$k$.
It holds that~$\FPT \subsetneq \XP{}$.
Intuitively, if a parameterized problem is in~$\XP{} \setminus \FPT{}$,
it is not likely to be efficiently solvable in practice.
Suppose, for example, that a problem is solvable in time~$n^{k}$
in the worst case.
Then already for~$n = 100$ and~$k = 10$, it could take ages to solve
this problem (see, e.g.,~\citealt{Downey12}).

\paragraph{Completeness Theory.}
Parameterized complexity also offers a \emph{completeness theory},
similar to the theory of \NP{}-completeness,
that provides a way to obtain evidence that
a parameterized problem is not fixed-parameter tractable.
Hardness for parameterized complexity classes is based on fpt-reductions,
which are many-one reductions where the parameter of one problem
maps into the parameter for the other.
More specifically, a parameterized problem~$Q$ is fpt-reducible to another
parameterized problem~$Q'$
if there is a mapping~$R$
that maps instances of~$Q$ to instances of~$Q'$ such that
(i)~$(I,k) \in Q$ if and only if~$R(I,k) = (I',k') \in Q'$,
(ii)~$k' \leq g(k)$ for a computable function~$g$, and
(iii)~$R$ can be computed in time~$f(k)\Card{I}^c$
for a computable function~$f$ and a constant~$c$.
A problem~$Q$ is \emph{hard} for a parameterized complexity
class~$\mtext{K}$ 
if every problem~$Q' \in \mtext{K}$ can be
fpt-reduced to~$Q$.
A problem~$Q$ is \emph{complete} for a parameterized
complexity class~$\mtext{K}$
if~$Q \in \mtext{K}$
and~$Q$ is $\mtext{K}$-hard.

Central to the completeness theory are the classes~$\W{1}
\subseteq \W{2} \subseteq \dotsc \subseteq \W{P} \subseteq \XP$
of the Weft hierarchy.
We will not define the classes~\W{\ensuremath{t}} in detail
(for details, see, e.g., \citealt{FlumGrohe06}).
It suffices to note that it is widely believed that~$\W{1} \neq \FPT$.%
\footnote{In fact, it holds that~$\W{1} \neq \FPT$, assuming that
$n$-variable 3SAT cannot be solved in subexponential time,
that is, in time~$2^{o(n)}$
\cite{ChenChorFellowsHuangJuedesKanjXia05,%
ChenKanj12,DowneyFellows13}.}
Thus, showing that a problem~$Q$ is \W{1}-hard gives evidence
that~$Q$ is not fpt-time solvable.

An example of a \W{1}-complete parameterized problem
is \Clique{} \cite{DowneyFellows95,DowneyFellows13}.
Instances for this problem consist of~$(G,k)$,
where~$G = (V,E)$ is an undirected graph, and~$k \in \NN$.
The parameter is~$k$, and the question is to decide
whether~$G$ contains a clique of size~$k$.

\paragraph{Para-K.}
For each classical complexity class~$\mtext{K}$,
we can construct a parameterized analogue~$\para{\mtext{K}}$
\cite{FlumGrohe03}.
Let~$\mtext{K}$ be a classical complexity class, e.g., \NP{}.
The parameterized complexity
class~$\para{\mtext{K}}$ is then defined as the class of all parameterized
problems~$L \subseteq \Sigma^{*} \times \NN{}$
for which there exist a computable function~$f : \NN{} \rightarrow \Sigma^{*}$
and a problem~$Q' \subseteq \Sigma^{*} \times \Sigma^{*}$
in~$K$, such
that for all instances~$(x,k) \in \Sigma^{*} \times \NN{}$
it holds that~$(x,k) \in Q$ if and only if~$(x,f(k)) \in Q'$.
Intuitively, the class~$\para{K}$ consists of all problems that are
in~$\mtext{K}$ after a precomputation that only involves the parameter.
A common example of such parameterized analogues of classical
complexity classes is the parameterized complexity class \para{\NP}.
Another example is~$\para{\P} = \FPT$.

If (the unparameterized variant of) a parameterized problem~$Q$
is in the class~$\mtext{K}$, then~$Q \in \para{\mtext{K}}$.
Also, if~$Q$ is already $\mtext{K}$-hard for a finite set of
parameter values, then~$Q$ is $\para{\mtext{K}}$-hard
\cite{FlumGrohe03}.

Using the classes \para{\mtext{K}} and the notion of fpt-reductions,
one can also provide evidence that certain parameterized problems
are not fixed-parameter tractable.
If a \para{\mtext{K}}-hard parameterized problem is fixed-parameter
tractable, then~$\mtext{K} = \P$.
For example, a \para{\NP}-hard parameterized problem is not
fixed-parameter tractable, unless~$\P = \NP$.

\paragraph{Para-NP and para-co-NP.}
The classes \para{\NP} and \para{\co{\NP}} are parameterized analogues
of the classes \NP{} and \co{\NP}.
The class \para{\NP} can alternatively be defined as the class of parameterized
problems that are solvable in fpt-time by a non-deterministic
algorithm~\cite{FlumGrohe03}.
Similarly, \para{\co{\NP}} can be defined using fpt-algorithms using
universal nondeterminism---%
i.e., nondeterministic fpt-algorithms that reject the input if at least
one sequence of nondeterministic choices leads the algorithm to reject.
It holds that~$\W{1} \subseteq \W{2} \subseteq \dotsm \subseteq \W{P}
\subseteq \para{\NP}$.

Another alternative definition of the class \para{\NP}---that can
be motivated by the amazing practical performance of SAT solving
algorithms (see, e.g.,~\citealt{BiereHeuleMaarenWalsh09})---%
is using the following parameterized variant of the
propositional satisfiability problem
\cite{DeHaan16,DeHaanSzeider14b,DeHaanSzeider14,%
DeHaanSzeider17}.
Let~$\SAT_{1} = \SB (\varphi,1) \SM \varphi \in \SAT \SE$ be the
problem SAT with a constant parameter~$k=1$.
The class \para{\NP} consists of all problems that can be fpt-reduced
to~$\SAT_{1}$.
In other words, \para{\NP} can be seen as the class of all parameterized
problems that can be solved by (1)~a fixed-parameter tractable encoding
into SAT, and (2)~using a SAT solving algorithm to then decide the problem.
The class \para{\co{\NP}} can be characterized in a similar way,
using UNSAT instead of SAT.
Consequently, problems in \para{\co{\NP}} can also be solved using the
combination of an fpt-time encoding and a SAT solving algorithm.

\paragraph{Para-PSPACE.}

The class \para{\PSPACE}
can alternatively be defined as the class
of all parameterized problems~$Q$
for which there exists a (deterministic or nondeterministic) algorithm
deciding whether~$(x,k) \in Q$
using space~$f(k) |x|^{O(1)}$, for some computable function~$f$.
It holds that~$\para{\NP} \cup \para{\co{\NP}} \subseteq \para{\PSPACE}$.

Another alternative characterization of \para{\PSPACE} is using
a parameterized variant of \TQBF{}---the problem of deciding
whether a given quantified Boolean formula is true.
Let~$\TQBF_{1} = \SB (\varphi,1) \SM \varphi \in \TQBF \SE$ be the
problem \TQBF{} with a constant parameter~$k=1$.
The class \para{\PSPACE} consists of all problems that can be fpt-reduced
to~$\TQBF_{1}$.
In other words, \para{\PSPACE} can be seen as the class of all parameterized
problems that can be solved by (1)~an fpt-time encoding
into TQBF, and (2)~using a TQBF solver to then decide the problem
(see, e.g.,~\citealt{BiereHeuleMaarenWalsh09}).

Yet another characterization of \para{\PSPACE} uses alternating Turing
machines (ATMs).
An ATM is a nondeterministic Turing machine where the states are partitioned
into existential and universal states (see, e.g.,~\citealt{FlumGrohe06}, Appendix~A.1).
A configuration of the ATM with an existential state is accepting if at least one
successor configuration is accepting, and a configuration with a universal state
is accepting if all successor configurations are accepting.
Intuitively, an ATM can alternate between existential and universal nondeterminism.
The class \para{\PSPACE} consists of all parameterized problems that can be
decided by an ATM in fixed-parameter tractable time.

\paragraph{Para-EXPTIME.}

The class \para{\EXPTIME}
can be defined as the class
of all parameterized problems~$Q$
for which there exists a deterministic algorithm deciding whether~$(x,k) \in Q$
in time~$f(k) 2^{|x|^{O(1)}}$, for some computable function~$f$.
It holds that~$\para{\PSPACE} \subseteq \para{\EXPTIME}$
and that~$\XP \subseteq \para{\EXPTIME}$.

For an overview of all parameterized complexity classes
that feature in this paper---and their relation---%
see Figure~\ref{fig:parameterized-landscape}.

\begin{figure}[htp!]
\begin{center}
\begin{tikzpicture}
  \node[] (fpt) at (0,-0.15) {$\FPT$};
  \node[] (w1) at (-1.25,0.5) {$\W{1}$};
  \node[] (cow1) at (1.25,0.5) {$\co{\W{1}}$};
  \node[] (paranp) at (-2.5,1.25) {$\para{\NP}$};
  \node[] (paraconp) at (2.5,1.25) {$\para{\co{\NP}}$};
  \node[] (xp) at (-1.25,1.25) {\phantom{p}\!\!\!\!$\XP{}$};
  \node[] (parapspace) at (0,2.2) {$\para{\PSPACE}$};
  \node[] (paraexptime) at (0,3.1) {$\para{\EXPTIME}$};
  \draw[->] (fpt) -- (w1);
  \draw[->] (fpt) -- (cow1);
  \draw[->] (w1) -- (paranp);
  \draw[->] (cow1) -- (paraconp);
  \draw[->] (paranp) -- (parapspace);
  \draw[->] (paraconp) -- (parapspace);
  \draw[->] (w1) -- (xp);
  \draw[->] (cow1) -- (xp);
  \draw[->] (xp) edge[bend left=35] (paraexptime);
  \draw[->] (parapspace) -- (paraexptime);
\end{tikzpicture}
\end{center}
\vspace{-10pt}
\caption{An overview of the landscape of parameterized
complexity classes that play a role in this paper.}
\label{fig:parameterized-landscape}
\end{figure}
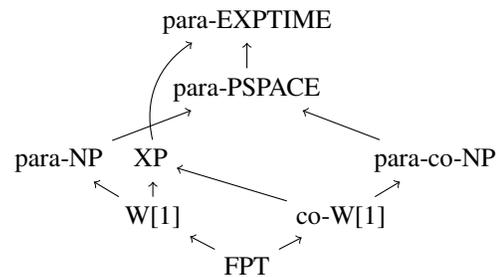

\section{Case Study 1: Concept Satisfiability for \ALC{} with respect to Nearly Acyclic TBoxes}
\label{sec:alc-almost-acyclic}

In this section, we provide our first case study to illustrate how
parameterized complexity can be used to obtain a more detailed
image of the computational complexity of description logic reasoning.
In particular, we consider the problem of concept satisfiability
for the description logic \ALC{} with respect to general TBoxes.
This problem is \EXPTIME{}-complete in general.
We consider two parameters for this problem.
One of these parameters does not help to reduce the complexity
of the problem---that is,
for this parameter the problem is \para{\EXPTIME}-complete.
The other of the two parameters does help to reduce the complexity
of the problem---that is,
for this parameter the problem is \para{\PSPACE}-complete.

We begin by revisiting the description logic \ALC{},
the problem of concept satisfiability with respect
to acyclic and general TBoxes,
and classical complexity results for this problem.
We then discuss our parameterized complexity results,
and how to interpret these results.

\subsection{The Description Logic \ALC{}}

Let~$N_C$,~$N_R$ and~$N_O$ be 
sets of
\emph{atomic concepts},
\emph{roles}, and \emph{individuals},
respectively. The triple~$(N_C,N_R,N_O)$ is called the \emph{signature}.
(We will often omit the signature if this is clear from the context.)

Concepts~$C$ are defined by the following grammar
in Backus-Naur form, for~$R \in N_R$ and~$A \in N_C$:
\[ C := A\ |\ \top\ |\ \bot\ |\ \neg C\ |\ C \sqcap C\ |\ C \sqcup C\ |\ \exists R. C\ |\ \forall R. C. \]

An \emph{interpretation~$\III = (\Delta^{\III},\cdot^{\III})$}
over a signature~$(N_C,N_R,N_O)$ consists of
a non-empty set~$\Delta^{\III}$ called the \emph{domain},
and an interpretation function~$\cdot^{\III}$ that maps
(1)~every individual~$a \in N_O$ to an
element~$a^{\III} \in \Delta^{\III}$,
(2)~every concept~$C$ to a subset of~$\Delta^{\III}$,
and (3)~every role~$R \in N_R$ to a subset
of~$\Delta^{\III} \times \Delta^{\III}$,
such that:\\[3pt]
\begin{minipage}{0.40\linewidth}
\begin{itemize}
  \item $\top^{\III} = \Delta^{\III}$; $\bot^{\III} = \emptyset$;
  \item $(\neg C)^{\III} = \Delta^{\III} \setminus C^{\III}$;
\end{itemize}
\end{minipage}
\begin{minipage}{0.58\linewidth}
\begin{itemize}
  \item $(C_1 \sqcap C_2)^{\III} = (C_1)^{\III} \cap (C_2)^{\III}$;
  \item $(C_1 \sqcup C_2)^{\III} = (C_1)^{\III} \cup (C_2)^{\III}$;
\end{itemize}
\end{minipage}
\begin{itemize}
  \item $(\exists R. C)^{\III} = \SB x \in \Delta^{\III} \SM$
    there exists some~$y \in C^{\III}$ such that $(x,y) \in R^{\III}$ $\SE$; and
  \item $(\forall R. C)^{\III} = \SB x \in \Delta^{\III} \SM$
    for each~$y$ such that $(x,y) \in R^{\III}$
    it holds that~$y \in C^{\III}$ $\SE$.
\end{itemize}

A \emph{general concept inclusion (GCI)} is a statement of the
form~$C \sqsubseteq D$, where~$C,D$ are concepts.
We write~$\III \models C \sqsubseteq D$
(and say that~$\III$ satisfies~$C \sqsubseteq D$)
if~$C^{\III} \subseteq D^{\III}$.
A \emph{(general) TBox~$\TTT$} is a finite set of GCIs.
A \emph{concept definition} is a statement of the
form~$A \equiv C$, where~$A \in N_C$ is an atomic concept,
and~$C$ is a concept.
We write~$\III \models A \equiv C$
(and say that~$\III$ satisfies~$A \equiv C$) if~$A^{\III} = C^{\III}$.
An \emph{acyclic TBox~$\TTT$} is a finite set of concept definitions
such that (1)~$\TTT$ does not contain two different concept
definitions~$A \equiv C_1$ and~$A \equiv C_2$ for any~$A \in N_C$,
and (2)~$\TTT$ contains no (direct or indirect) cyclic definitions---%
that is, the graph~$G_{\TTT}$ with vertex set~$N_C$ that contains
an edge~$(A,B)$ if and only if~$\TTT$ contains a concept
definition~$A \equiv C$ where~$B$ occurs in~$C$ is acyclic.
An interpretation~$\III$ satisfies a (general or acyclic) TBox~$\TTT$
if~$\III$ satisfies all GCIs or concept definitions in~$\TTT$.

A \emph{concept assertion} is a statement of the form~$C(a)$,
where~$a \in N_O$ and~$C$ is a concept.
A \emph{role assertion} is a statement of the form~$R(a,b)$,
where~$a,b \in N_O$ and~$R \in N_R$.
We write~$\III \models C(a)$ (and say that~$\III$ satisfies~$C(a)$)
if~$a^{\III} \in C^{\III}$.
Moreover, we write~$\III \models R(a,b)$ (and say that~$\III$ satisfies~$R(a,b)$)
if~$(a^{\III},b^{\III}) \in R^{\III}$.
An \emph{ABox~$\AAA$} is a finite set of concept and role assertions.

\subsection{Classical Complexity Results}

An important reasoning problem for description logics is the problem
of \emph{concept satisfiability}.
In this decision problem, the input consists of a concept~$C$
and a TBox~$\TTT$, and the question is whether~$C$ is satisfiable
with respect to~$\TTT$---%
that is, whether there exists an interpretation~$\III$ such that~$\III \models \TTT$
and~$C^{\III} \neq \emptyset$.
The problem of concept satisfiability is \PSPACE{}-complete, both for the case
where~$\TTT$ is empty and for the case where~$\TTT$ is an acyclic TBox.
For the case where~$\TTT$ is a general TBox, the problem
is \EXPTIME{}-complete.

\begin{proposition}[\citealt{DoniniMassacci00};~\citealt{Schild91}]
\label{prop:alc-exptime}
Concept satisfiability for the logic \ALC{} with respect to general TBoxes
is \EXPTIME{}-complete.
\end{proposition}

\begin{proposition}[\citealt{BaaderLutzMilicicSattlerWolter05};~\citealt{SchmidtSchaussSmolka91}]
\label{prop:alc-pspace-acyclic}
Concept satisfiability for the logic \ALC{} with respect to acyclic TBoxes
is \PSPACE{}-complete.
\end{proposition}

\subsection{Parameterized Complexity Results}

We consider a parameterized variant of the problem of concept satisfiability
for \ALC{} where the parameter captures the distance towards acyclicity for
the given TBox.
That is, for this parameterized problem, the input consists of a concept~$C$,
an acyclic TBox~$\TTT_1$, and a general TBox~$\TTT_2$.
The parameter is~$k = |\TTT_2|$, and the question is whether~$C$ is
satisfiable with respect to~$\TTT_1 \cup \TTT_2$---%
that is, whether there exists an interpretation~$\III$ such that~$\III \models \TTT_1$,~%
$\III \models \TTT_2$, and~$C^{\III} \neq \emptyset$.

Parameterizing by the size of~$\TTT_2$ does not offer an improvement
in the complexity of the problem---that is, this parameter leads to
\para{\EXPTIME}-completeness.

\begin{theorem}
\label{thm:alc-para-exptime}
Concept satisfiability for \ALC{} with respect to both an acyclic TBox~$\TTT_1$
and a general TBox~$\TTT_2$ is \para{\EXPTIME}-complete
when parameterized by~$|\TTT_2|$. 
\end{theorem}
\begin{proof}
Membership in \para{\EXPTIME} follows from the fact that the unparameterized
version of the problem is in \EXPTIME{}.
To show \para{\EXPTIME}-hardness, it suffices to show that the problem is already
\EXPTIME{}-hard for a constant value of the parameter \cite{FlumGrohe03}.
We do so by giving a reduction from the problem of concept satisfiability
for \ALC{} with respect to general TBoxes.

Let~$C$ be a concept and let~$\TTT$ be a general TBox.
Moreover, let~$\TTT = \SBs C_1 \sqsubseteq D_1,\dotsc,C_m \sqsubseteq D_m \SEs$.
We construct an acyclic TBox~$\TTT_1$
and a general TBox~$\TTT_2$ such that~$C$ is satisfiable with respect to~$\TTT$
if and only if it is satisfiable with respect to~$\TTT_1 \cup \TTT_2$.
Let~$A$ be a fresh atomic concept.
We let~$\TTT_1 = \SBs A \equiv \bigsqcap_{i=1}^{m} (\neg C_i \sqcup D_i) \SEs$,
and we let~$\TTT_2 = \SBs \top \sqsubseteq A \SEs$.
It is straightforward to verify that~$C$ is satisfiable with respect to~$\TTT$
if and only if it is satisfiable with respect to~$\TTT_1 \cup \TTT_2$.
Moreover,~$|\TTT_2|$ is constant.
From this, we can conclude that the problem is \para{\EXPTIME}-hard.
\end{proof}

Intuitively, restricting only the number (and size) of the general TBox~$\TTT_2$
does not restrict the problem, as we can encode a general TBox of arbitrary
size in the acyclic TBox (together with a small general TBox).
If we restrict the number of concepts impacted by the general TBox, however,
we do get an improvement in the complexity of the problem.

Let~$\TTT_1$ be an acyclic TBox and let~$\TTT_2$ be a general TBox.
We define the set of \emph{concepts impacted by~$\TTT_2$ (w.r.t.~$\TTT_1$)}
as the smallest set~$I$ of concepts that is closed under (syntactic)
subconcepts and that satisfies that
(A)~whenever~$C \sqsubseteq D \in \TTT_2$,
then~$C,D \in I$, and
(B)~whenever~$A \in I$ and~$A \equiv C \in \TTT_1$,
then~$C \in I$.
%
If we parameterize the problem of concept satisfiability with respect to both
an acyclic TBox~$\TTT_1$ and a general TBox~$\TTT_2$ by the number
of concepts impacted by~$\TTT_2$, the complexity of the problem jumps
down to \para{\PSPACE}.

\begin{theorem}
\label{thm:alc-para-pspace}
Concept satisfiability for \ALC{} with respect to both an acyclic TBox~$\TTT_1$
and a general TBox~$\TTT_2$ is \para{\PSPACE}-complete when
parameterized by the number~$k$ of concepts
that are impacted by~$\TTT_2$ (w.r.t.~$\TTT_1$).
\end{theorem}
\begin{proof}
Hardness for \para{\PSPACE} follows directly from the fact that the problem
is already \PSPACE{}-hard when~$\TTT_2$ is empty
(Proposition~\ref{prop:alc-pspace-acyclic})---and thus the number of concepts
impacted by~$\TTT_2$ is~$0$.
We show membership in \para{\PSPACE} by exhibiting a
nondeterministic algorithm to solve the problem that runs in
space~$f(k) \cdot n^{O(1)}$, for some computable function~$f$.
Let~$C$ be a concept, let~$\TTT_1$ be an acyclic TBox,
and let~$\TTT_2$ be a general TBox.
%
We may assume without loss of generality that all concepts
occurring in~$\TTT_1$ are in negation normal form---%
that is, negations occur only directly in front of atomic concepts.
If this were not the case, we could straightforwardly
transform~$\TTT_1$ to a TBox that does have this property
in polynomial time, by introducing new atomic concepts~$A'$
for any negated concept~$\neg A$.

The algorithm that we use is the usual tableau
algorithm (with static blocking) for \ALC{}---see, e.g.,~\cite{BaaderSattler01}.
That is, it aims to construct a tree that can be used to construct
an interpretation satisfying~$C$,~$\TTT_1$ and~$\TTT_2$.
For each node in the tree, it first exhaustively applies the rules
for the~$\sqcup$ and~$\sqcap$ operators,
the rules for the concept definitions in~$\TTT_1$,
and the rules for the GCIs in~$\TTT_2$
(for each~$C \sqsubseteq D \in \TTT_2$ adding the
concept~$\neg C \sqsubseteq D$ to a node),
before applying the rules for the~$\exists$ and~$\forall$ operators.
Moreover, it applies all rules exhaustively to one node of the tree
before moving to another node.
Additionally, the algorithm uses the following usual
blocking condition (subset blocking):
the rule for the~$\exists$ operator cannot be
applied to a node~$x$ that has a predecessor~$y$ in the tree
that is labelled with all concepts that~$x$ is labelled with
(and possibly more).
It is straightforward to verify that this tableau algorithm correctly
decides the problem.

We argue that this algorithm requires space~$2^{k} \cdot n^{O(1)}$,
where~$k$ is the number of concepts impacted by~$\TTT_2$
and~$n$ denotes the input size.
%
%
It is straightforward to verify that there is a polynomial~$p$ such
that each node in the tree constructed by
the tableau algorithm that is more than~$p(n)$ steps
away from the root of the tree is only labelled with concepts that
are impacted by~$\TTT_2$.
Since there are only~$k$ concepts that are impacted by~$\TTT_2$,
we know that in each branch of the tree, the blocking condition
applies at depth at most~$2^k \cdot p(n)$, and thus that each
branch is of length at most~$2^k \cdot p(n)$.
From this, it follows that this algorithm requires
space~$2^k \cdot n^{O(1)}$,
and thus that the problem is in \para{\PSPACE}.
\end{proof}

\subsection{Interpretation of the Results}

The results in this section are summarized
in Table~\ref{table:alc-concept-sat-tboxes}.
The parameterized results of Theorems~\ref{thm:alc-para-exptime}
and~\ref{thm:alc-para-pspace} show that parameterized complexity theory
can make a distinction between the complexity
of the two variants of the problem that classical complexity
theory is blind to.
Classically, both variants are \EXPTIME{}-complete,
but one parameter can be used to get a polynomial-space
algorithm (when an additional~$2^k$ factor for the parameter~$k$),
whereas the other parameter requires exponential space,
no matter what additional~$f(k)$ factor is allowed.
The \para{\PSPACE} result 
of Theorem~\ref{thm:alc-para-pspace} also
yields an algorithm
solving the problem using (1)~an fpt-time encoding into the
problem TQBF, and then (2)~using a TQBF solver
to decide the problem
(see, e.g.,~\citealt{BiereHeuleMaarenWalsh09}).

\begin{table}[!htb]
  \centering
  
  \begin{small}
  \begin{tabular}{@{\ \ } p{3.0cm} @{\quad}|@{\quad}p{4.3cm} @{\ \ } } \toprule
      \multirow{2}{*}{\textit{parameter}} & \textit{complexity of \ALC{} concept} \\
      & \textit{satisfiability w.r.t.~$\TTT_1$ and~$\TTT_2$} \\
      \midrule
      \hspace{0.5pt}-- & \EXPTIME{}-c \hfill (Proposition~\ref{prop:alc-exptime}) \\[3pt]
      $|\TTT_2|$ & \para{\EXPTIME}-c \hfill (Theorem~\ref{thm:alc-para-exptime}) \\[3pt]
      \# of concepts impacted\newline by~$\TTT_2$ (w.r.t.~$\TTT_1$) & \para{\PSPACE}-c \hfill (Theorem~\ref{thm:alc-para-pspace}) \\[10pt]
    \bottomrule
  \end{tabular}
  \end{small}
  
  \caption{The parameterized complexity of \ALC{} concept satisfiability
    w.r.t. both an acyclic TBox~$\TTT_1$ and
    a general TBox~$\TTT_2$, for different parameters.}
  \label{table:alc-concept-sat-tboxes}
\end{table}

\section{Case Study 2: Concept Satisfiability for \ALC{}, \ALE{}, \ALU{} and \AL{}}

In this section, we provide our second case study to illustrate how
parameterized complexity can be used to obtain a more detailed
image of the computational complexity of description logic reasoning.
In particular, we consider the problem of concept satisfiability
for the description logic \ALC{}.
This problem is \PSPACE{}-complete in general.
We consider several parameters that measure the distance
to the logics \ALE{} and \ALU{}.
The logics \ALE{} and \ALU{} are obtained from \ALC{} by
disallowing concept union and full existential qualification,
respectively.
The parameters that we consider both help to reduce the complexity
of the problem.
One parameter renders the problem \para{\co{\NP}}-complete.
The other parameter renders the problem \para{\NP}-complete.
The combination of both parameters renders the problem
fixed-parameter tractable.

We begin by revisiting the description logics \ALE{} and \ALU{}
(and their intersection \AL{}),
and classical complexity results for the problem of concept
satisfiability for these logics.
We then discuss our parameterized complexity results,
and how to interpret these results.

\subsection{The Description Logics \ALE{}, \ALU{} and \AL{}}

In order to obtain the description logics \ALE{}, \ALU{} and \AL{},
we consider a (syntactic) variant of the logic \ALC{} where all concepts
are in \emph{negation normal form}.
That is, negations only occur immediately followed by
atomic concepts.
Put differently, we consider concepts~$C$ that are 
defined
as follows, for~$R \in N_R$ and~$A \in N_C$:
\[ C := A\ |\ \neg A\ |\ \top\ |\ \bot\ |\ C \sqcap C\ |\ C \sqcup C\ |\ \exists R. C\ |\ \forall R. C. \]
One can transform any \ALC{} concept into
negation normal form in linear time
(see, e.g.,~\citealt{BaaderHorrocksSattler08}).
The semantics of this variant of \ALC{} is defined exactly as
described in the previous section.
Throughout this section, we will only consider this variant of \ALC{}.

The description logic \ALE{} is obtained from the logic \ALC{}
by forbidding any occurrence of the operator~$\sqcup$.
%
The description logic \ALU{} is obtained from the logic \ALC{}
by requiring that for every occurrence~$\exists R. C$ of the
existential quantifier it holds that~$C = \top$; that is,
only limited existential quantification~$\exists R. \top$ is allowed.
%
The description logic \AL{} contains those concepts
that are concepts in both \ALE{} and \ALU{}---that is,
\AL{} is the intersection of \ALE{} and \ALU{}.

Thus, the logics \ALE{} and \ALU{} are obtained from \ALC{}
by means of two orthogonal restrictions:
disallowing concept union
and replacing full existential qualification by limited existential quantification,
respectively.
The logic \AL{} is obtained from \ALC{}
by using both of these 
restrictions.

\subsection{Classical Complexity Results}

In this section, we consider the problem
of \emph{concept satisfiability} with respect to empty TBoxes.
In this decision problem, the input consists of a concept~$C$,
and the question is whether~$C$ is satisfiable---%
that is, whether there exists an interpretation~$\III$
such that~$C^{\III} \neq \emptyset$.
This problem is \PSPACE{}-complete for \ALC{},
\co{\NP}-complete for \ALE{}, \NP{}-complete for \ALU{},
and polynomial-time solvable for \AL{}.

\begin{proposition}[\citealt{SchmidtSchaussSmolka91}]
\label{prop:alc-pspace}
Concept satisfiability for the logic \ALC{} is \PSPACE{}-complete.
\end{proposition}

\begin{proposition}[\citealt{DoniniHollunderLenzeriniSpaccamelaNardiNutt02}]
\label{prop:ale-conp}
Concept satisfiability for the logic \ALE{}
is \co{\NP}-complete.
\end{proposition}

\begin{proposition}[\citealt{DoniniLenzeriniNardiNutt97}]
\label{prop:alu-np}
Concept satisfiability for the logic \ALU{}
is \NP-complete.
\end{proposition}

\begin{proposition}[\citealt{SchmidtSchaussSmolka91}]
\label{prop:al-p}
Concept satisfiability for the logic \AL{}
is polynomial-time solvable.
\end{proposition}

\begin{table*}[ht!]
\begin{center}
  \begin{small}
  \begin{tabular}{p{1.75cm} p{12.4cm}}
    \toprule
    \multicolumn{2}{l}{\textbf{The $\sqcap$-rule}} \\
    \textit{Condition} & $\AAA$ contains~$(C_1 \sqcap C_2)(x)$,
      but it does not contain both~$C_1(x)$ and~$C_2(x)$. \\
    \textit{Action} & $\AAA' = \AAA \cup \SBs C_1(x), C_2(x) \SEs$. \\
    \midrule
    \multicolumn{2}{l}{\textbf{The $\sqcup$-rule}} \\
    \textit{Condition} & $\AAA$ contains~$(C_1 \sqcup C_2)(x)$,
      but neither~$C_1(x)$ nor~$C_2(x)$. \\
    \textit{Action} & $\AAA' = \AAA \cup \SBs C_1(x) \SEs$,
      $\AAA'' = \AAA \cup \SBs C_2(x) \SEs$. \\
    \midrule
    %
    %
    \multicolumn{2}{l}{\textbf{The $\exists$-rule}} \\
    \textit{Condition} & $\AAA$ contains~$(\exists R.C)(x)$,
      but there is no individual~$z$ such
      that~$C(z)$ and~$R(x,z)$ are in~$\AAA$. \\
    \textit{Action} & $\AAA' = \AAA \cup \SBs C(y), R(x,y) \SEs
      \setminus \SB (\exists R'. C')(x') \in \AAA \SM x' = x,
      (\exists R'.C') \neq (\exists R.C) \SE$, \\
      & where~$y$ is an arbitrary individual not occurring
      in~$\AAA$. \\
    \midrule
    \multicolumn{2}{l}{\textbf{The $\forall$-rule}} \\
    \textit{Condition} & $\AAA$ contains~$(\forall R.C)(x)$
      and~$R(x,y)$,
      but it does not contain~$C(y)$. \\
    \textit{Action} & $\AAA' = \AAA \cup \SBs C(y) \SEs$. \\
    \midrule
    \multicolumn{2}{l}{\textbf{The $\bot$-rule}} \\
    \textit{Condition} & $\AAA$ contains~$A(x)$
      and~$(\neg A)(x)$,
      but it does not contain~$\bot$. \\
    \textit{Action} & $\AAA' = \AAA \cup \SBs \bot \SEs$. \\
    \bottomrule
  \end{tabular}
  \end{small}
\end{center}
  \vspace{-5pt}
  \caption{Transformation rules of the tableau algorithm for \ALC{} concept satisfiability.}
  \label{table:tableau-rules}
  \vspace{-5pt}
\end{table*}

\subsection{Parameterized Complexity Results}

In order to conveniently describe the parameterized complexity results
that we will establish in this section, we firstly describe an algorithm
for deciding concept satisfiability for \ALC{} in polynomial space
(see, e.g.,
\citealt{BaaderCalvaneseMcGuinnessNardiPatelSchneider03},
Chapter 2).
To use this algorithm to prove the parameterized
complexity results in this section, we describe a variant of the
algorithm that can be implemented by a polynomial-time
alternating Turing machine---i.e., a nondeterministic Turing machine
that can alternate between existential and universal nondeterminism.

The algorithm uses ABoxes~$\AAA$ as data structures, and works
by extending these ABoxes by means of several transformation
rules.
These rules are described Table~\ref{table:tableau-rules}---%
however, not all rules are applied in the same fashion.
The $\sqcap$-rule, the $\forall$-rule and the $\bot$-rule
are used as deterministic rules, and are applied greedily
whenever they apply.
The $\sqcup$-rule and the $\exists$-rule are nondeterministic rules,
but are used in a different fashion.
The $\sqcup$-rule transforms an ABox~$\AAA$ into one of two
different ABoxes~$\AAA'$ or~$\AAA''$ nondeterministically.
The $\sqcup$-rule is implemented using existential nondeterminism---i.e.,
the algorithm succeeds if at least one choice of~$\AAA'$ and~$\AAA''$
ultimately leads to the algorithm accepting.
(For more details on existential and universal nondeterminism
and alternating Turing machines, see, e.g.,~\citealt{FlumGrohe06},
Appendix~A.1.)
The $\exists$-rule, on the other hand, transforms an ABox~$\AAA$ into
a unique next ABox~$\AAA'$, but it is a nonmonotonic rule that can
be applied in several ways---the condition can be instantiated in different
ways, and these instantiations are not all possible anymore after having
applied the rule.
The $\exists$-rule is implemented using universal nondeterminism---i.e.,
the algorithm succeeds if all ways of instantiating the condition of
the $\exists$-rule (and applying the rule accordingly) ultimately
lead to the algorithm accepting.

The tableau algorithm works as follows.
Let~$C_0$ be an \ALC{} concept for which we want to decide
satisfiability.
We construct an initial ABox~$\AAA_0 = \SBs C_0(x_0) \SEs$,
where~$x_0 \in N_{O}$ is an arbitrary individual.
We proceed in two alternating phases:~(I) and~(II)---%
starting with phase~(I).

In phase~(I), we apply the deterministic rules (the $\sqcap$-rule,
the $\forall$-rule and the $\bot$-rule) and
the nondeterministic $\sqcup$-rule exhaustively,
until none of these rules is applicable anymore.
For the $\sqcup$-rule we use existential nondeterminism to
choose which of~$\AAA'$ and~$\AAA''$ to use.
When none of these rules is applicable anymore,
we proceed to phase~(II).
In phase~(II), we apply the $\exists$-rule once, using
universal nondeterminism to choose how to instantiate the
condition (and we apply the rule accordingly).
Then, we go back to phase~(I).

Throughout the execution of the algorithm, there is always
a single current ABox~$\AAA$.
Whenever it holds that~$\bot \in \AAA$, the algorithm rejects.
If at some point no rule is applicable anymore---that is,
if at some point we are in phase~(II) and the $\exists$-rule
is not applicable---%
the algorithm accepts.

This algorithm essentially works the same way as known
tableau algorithms for \ALC{} concept satisfiability
(see, e.g.,
\citealt{BaaderCalvaneseMcGuinnessNardiPatelSchneider03},
Chapter 2).
The only difference is that in the algorithm described above
the implementation of the $\sqcup$-rule using existential
nondeterminism and the implementation of the $\exists$-rule
using universal nondeterminism is built in.
In the literature, typically descriptions of tableau algorithms leave
freedom for different implementations of the way in which the
search tree is traversed.
One can think of the algorithm described above as traversing
a search tree that is generated by the different (existential and
universal) nondeterministic choices that are made in the execution
of the algorithm.
This search tree is equivalent to the search tree of the usual tableau
algorithm for \ALC{} concept satisfiability.
Thus, we get that the algorithm is correct.
In fact, this algorithm is a reformulation
of the standard algorithm known from the literature
\cite{BaaderCalvaneseMcGuinnessNardiPatelSchneider03}.

\begin{proposition}[{\citealt{BaaderCalvaneseMcGuinnessNardiPatelSchneider03}}]
The tableau algorithm described above
for an alternating polynomial-time Turing machine
correctly decides concept satisfiability for \ALC{}.
\end{proposition}

We will now consider several parameterized variants
of the problem of concept satisfiability for \ALC{}.
These parameters, in a sense, measure the distance of
an \ALC{} concept to the logics \ALE{}, \ALU{} and \AL{},
respectively.
We will make use of the tableau algorithm described above
to establish upper bounds on the complexity of these problems.
Lower bounds follow directly from
Propositions~\ref{prop:ale-conp}--\ref{prop:al-p}.

We begin with the parameterized variant of \ALC{} concept
satisfiability where the parameter measures the distance to \ALE{}.

\begin{theorem}
\label{thm:alc-para-conp}
Concept satisfiability for the logic \ALC{},
parameterized by the number of occurrences
of the union operator~$\sqcup$ in~$C$,
is \para{\co{\NP}}-complete.
\end{theorem}
\begin{proof}
Hardness for \para{\co{\NP}} follows from the fact
that \ALE{} concept satisfiability is \co{\NP}-complete
(Proposition~\ref{prop:ale-conp}).
Any \ALE{} concept is an \ALC{} concept with zero occurrences
of the union operator~$\sqcup$.
Therefore, the problem of \ALC{} concept satisfiability
parameterized by the number~$k$ of occurrences
of the union operator~$\sqcup$ in~$C$
is already \co{\NP}-hard for the parameter value~$k = 0$.
From this, it follows that the parameterized problem
is \para{\co{\NP}}-hard \cite{FlumGrohe03}.

To show that the parameterized problem is also contained
in \para{\co{\NP}}, we describe an algorithm that can be implemented
by an alternating Turing machine that only makes use of universal
nondeterminism and that runs in fixed-parameter tractable time.
This algorithm is similar to the tableau algorithm for \ALC{}
described above, with the only difference that the $\sqcup$-rule is now
not implemented using existential nondeterminism.
Instead, we deterministically iterate over all possible choices that can
be made in executions of the $\sqcup$-rule.
That is, whenever the $\sqcup$-rule is applied, resulting in two possible
next ABoxes~$\AAA'$ and~$\AAA''$, we firstly continue the algorithm
with~$\AAA'$, and if the continuation of the algorithm with~$\AAA'$ failed,
we then continue the algorithm with~$\AAA''$ instead.

Let~$k$ be the number of occurrences of the union operator~$\sqcup$
in~$C$. For each occurrence, the $\sqcup$-rule is applied at most once.
Therefore, the total number of possible choices resulting from executions
of the $\sqcup$-rule is at most~$2^k$.
Therefore, this modification of the algorithm can be implemented by
an alternating Turing machine that only uses universal nondeterminism
and that runs in time~$2^k \cdot |C|^{O(1)}$.
In other words, the problem is in \para{\co{\NP}},
and thus is \para{\co{\NP}}-complete.
\end{proof}

\begin{theorem}
\label{thm:alc-para-np}
Concept satisfiability for the logic \ALC{},
parameterized by the number of occurrences
of full existential qualification~$\exists R. C$ in~$C$,
is \para{\NP}-complete.
\end{theorem}
\begin{proof}
Hardness for \para{\NP} follows from the fact
that \ALU{} concept satisfiability is \NP-complete
(Proposition~\ref{prop:alu-np}).
Any \ALU{} concept is an \ALC{} concept with zero occurrences
of full existential qualification~$\exists R. C$.
Therefore, the problem of \ALC{} concept satisfiability
parameterized by the number of occurrences
of full existential qualification~$\exists R. C$ in~$C$
is already \NP-hard for the parameter value~$k = 0$.
From this, it follows that the parameterized problem
is \para{\NP}-hard \cite{FlumGrohe03}.

To show membership in \para{\NP}, we modify the tableau
algorithm for \ALC{}, similarly to the way we did in the proof
of Theorem~\ref{thm:alc-para-conp}.
In particular, we describe an algorithm that can be implemented
by an alternating Turing machine that only makes use of existential
nondeterminism and that runs in fixed-parameter tractable time.
We do so by executing the $\exists$-rule deterministically, instead
of using universal nondeterminism.
That is, instead of using universal nondeterminism to choose which
instantiation of the condition of the $\exists$-rule to use, we iterate
over all possibilities deterministically.

Let~$k$ be the number of occurrences of full existential
quantification~$\exists R. C$ in~$C$.
At each point, there are at most~$k$ different ways of instantiating
the $\exists$-rule.
Moreover, after having applied the $\exists$-rule for at most~$k$
times, the $\exists$-rule is not applicable anymore.
Therefore, the total number of possible choices to iterate over
is at most~$k^k$.
Therefore, this modification of the algorithm can be implemented by
an alternating Turing machine that only uses existential nondeterminism
and that runs in time~$k^k \cdot |C|^{O(1)}$.
In other words, the problem is in \para{\NP},
and thus is \para{\NP}-complete.
\end{proof}

\begin{theorem}
\label{thm:alc-fpt}
Concept satisfiability for the logic \ALC{},
parameterized by both (i)~the number of occurrences
of the union operator~$\sqcup$ in~$C$
and (ii)~the number of occurrences
of full existential qualification~$\exists R. C$ in~$C$,
is fixed-parameter tractable.
\end{theorem}
\begin{proof}[Proof (sketch)]
We can modify the alternating polynomial-time
tableau algorithm for \ALC{} concept satisfiability
to work in deterministic fpt-time
by implementing both the $\sqcup$-rule and
the $\exists$-rule deterministically, iterating sequentially
over all possible choices that can be made for these rules.
That is, we combine the ideas behind the proofs of
Theorems~\ref{thm:alc-para-conp} and~\ref{thm:alc-para-np}.
We omit the details of this fpt-time algorithm.
\end{proof}
%

\subsection{Interpretation of the Results}

The results in this section are summarized
in Table~\ref{table:alc-concept-sat}.
Similarly as for the first case study,
the results for the second case study
show that parameterized complexity theory
can make distinctions that classical complexity theory does not see.
The problems studied in Theorems~\ref{thm:alc-para-conp},%
~\ref{thm:alc-para-np} and~\ref{thm:alc-fpt} are all \PSPACE{}-complete
classically, yet from a parameterized point of view their
complexity goes down to \para{\NP}, \para{\co{\NP}} and \FPT{}.
The \para{\NP}- and \para{\co{\NP}}-completeness results
of Theorems~\ref{thm:alc-para-conp} and~\ref{thm:alc-para-np}
also yield algorithms that (1)~firstly use an fpt-encoding
to an instance of SAT and (2)~then use a SAT solver to decide
the problem (see, e.g.,~\citealt{BiereHeuleMaarenWalsh09}).

\begin{table}[!h]
  \centering
  
  \begin{small}
  \begin{tabular}{@{\ \ } p{3.5cm} @{\quad}|@{\quad}p{3.8cm} @{\ \ } } \toprule
      \multirow{2}{*}{\textit{parameter}} & \textit{complexity of \ALC{}} \\
      & \textit{concept satisfiability} \\
      \midrule
      \hspace{0.5pt}-- & \PSPACE{}-c \hfill (Proposition~\ref{prop:alc-pspace}) \\[3pt]
      \# of occurrences of~$\sqcup$ & \para{\co{\NP}}-c \hfill (Theorem~\ref{thm:alc-para-conp}) \\[3pt]
      \# of occurrences of~$\exists$ & \para{\NP}-c \hfill (Theorem~\ref{thm:alc-para-np}) \\[3pt]
      \# of occurrences of~$\sqcup$ and~$\exists$ & FPT \hfill (Theorem~\ref{thm:alc-fpt}) \\[3pt]
    \bottomrule
  \end{tabular}
  \end{small}
  
  \caption{The parameterized complexity of \ALC{} concept satisfiability
    (with no TBoxes) for different parameters.}
  \label{table:alc-concept-sat}
\end{table}

\section{Case Study 3: A Parameterized Complexity View on Data Complexity}

%
%

In this section, we provide our third case study illustrating
the use of parameterized complexity for the analysis
of description logic reasoning.
This third case study is about refining the complexity analysis for
cases where one part of the input is much smaller than another part.
Typically, these cases occur where there is a small TBox and a small
query, but where there is a large database of facts (in the form of an ABox).
What is often done is that the size of the TBox and the query are seen
as fixed constants---and the complexity results are grouped under the
name of ``data complexity.''

In this section, we will look at two concrete polynomial-time data complexity
results for the description logic \ELI{}.
Even though the data complexity view gives the same
outlook on the complexity of these problems,
we will use the viewpoint of parameterized complexity theory to argue
that these two problems in fact have a different complexity.
One of these problems is more efficiently solvable than the other.

We chose the example of \ELI{} to illustrate our point
because it is technically straightforward.
More intricate fixed-parameter tractability results
for conjunctive query answering in description logics
have been obtained in the literature
\cite{BienvenuKikotKontchakovPodolskiiRyzhikovZakharyaschev17,%
BienvenuKikotKontchakovRyzhikovZakharyaschev17,%
KikotKontchakovZakharyaschev11}.

We begin by reviewing the description logic \ELI{},
and the two reasoning problems for this logic that we will
look at (instance checking and conjunctive query entailment).
We will review the classical complexity results for these two problems,
including the data complexity results.
We will then use results from the literature to
give a parameterized complexity analysis for these two
problems, and argue why the parameterized complexity perspective gives
a more accurate view on the complexity of these problems.

\subsection{The Description Logic \ELI{}}

To define the logic \ELI{}, we first consider the logic \EL{}.
The description logic \EL{} is obtained from the logic \ALC{} by forbidding
any use of the negation operator ($\neg$),
the empty concept ($\bot$),
the union operator ($\sqcup$),
and universal quantification ($\forall R. C$).
%
The description logic \ELI{} is obtained from the logic \EL{} by
introducing \emph{inverse roles}.
That is, \ELI{} concepts are defined by the following grammar
in Backus-Naur form, for~$R \in N_R$ and~$A \in N_C$:
\[ C := A\ |\ \top\ |\ C \sqcap C\ |\ \exists R. C\ |\ \exists R^{-}. C. \]
Interpretations~$\III = (\Delta^{\III},\cdot^{\III})$ for \ELI{} 
are defined as interpretations for \EL{}
with the following addition:
\begin{itemize}
  \item $(\exists R^{-}. C)^{\III} = \SB x \in \Delta^{\III} \SM$
    there exists some~$y \in C^{\III}$ such that $(y,x) \in R^{\III}$ $\SE$.
\end{itemize}

\subsection{Classical Complexity Results}

We consider two reasoning problems for the logic \ELI{}.
The first problem that we consider is the problem of
\emph{instance checking}.
In this problem, the input consists of an ABox~$\AAA$,
a (general) TBox~$\TTT$, an individual name~$a$ and a concept~$C$,
and the question is whether~$\AAA,\TTT \models C(a)$---%
that is, whether for each interpretation~$\III$ such that~$\III \models \AAA$
and~$\III \models \TTT$ it holds that~$\III \models C(a)$.

The second problem that we consider is the problem
of \emph{conjunctive query entailment} (which can be seen as a
generalization of the problem of instance checking).
A \emph{conjunctive query} is a set~$q$ of atoms
of the form~$C(v)$ and~$R(u,v)$,
where~$C$ is a concept, where~$R \in N_R$,
and where~$u,v$ are \emph{variables}.
Let~$\Var{q}$ denote the set of variables occurring in~$q$.
Let~$\III = (\Delta^{\III},\cdot^{\III})$ be an interpretation
and let~$\pi$ be a mapping from~$\Var{q}$ to~$\Delta^{\III}$.
We write~$\III \models^{\pi} C(v)$ if~$\pi(v) \in C^{\III}$,
we write~$\III \models^{\pi} R(u,v)$ if~$(\pi(u),\pi(v)) \in R^{\III}$,
we write~$\III \models^{\pi} q$ if~$\III \models^{\pi} \alpha$
for all~$\alpha \in q$, and we
write~$\III \models q$ if~$\III \models^{\pi} q$
for some~$\pi : \Var{q} \rightarrow \Delta^{\III}$.
For any ABox~$\AAA$ and TBox~$\TTT$, we write~$\AAA,\TTT \models q$
if~$\III \models q$ for each interpretation~$\III$ such
that~$\III \models \AAA$ and~$\III \models \TTT$.
In the problem of conjunctive query entailment,
the input consists of a (general) TBox~$\TTT$,
an ABox~$\AAA$, and a conjunctive query~$q$,
and the question is to decide whether~$\AAA,\TTT \models q$.

Both the problem of instance checking and the problem
of conjunctive query entailment for \ELI{}
are \EXPTIME{}-complete in general.

\begin{proposition}[\citealt{BaaderBrandtLutz05};~\citeyear{BaaderBrandtLutz08}]
\label{prop:eli-ic-exp}
Instance checking for \ELI{} is \EXPTIME{}-complete.
\end{proposition}

\begin{corollary}[\citealt{BaaderBrandtLutz05};~\citeyear{BaaderBrandtLutz08}]
\label{cor:eli-cqe-exp}
Conjunctive query entailment for \ELI{} is \EXPTIME{}-complete.
\end{corollary}

The results of Propositions~\ref{prop:eli-ic-exp}
and Corollary~\ref{cor:eli-cqe-exp} are typically called ``combined complexity'' results---%
meaning that all elements of the problem statement are given as
inputs for the problem.
To study how the complexity of these problems increases when the size of
the ABox~$\AAA$ grows---and when the size of the TBox~$\TTT$ and the
size of the query~$q$ remain the same---often different variants of the
problems are studied.
In these variants, the TBox~$\TTT$ and the query~$q$ are fixed
(and thus not part of the problem input), and only the ABox~$\AAA$ is given
as problem input.
That is, there is a variant of the problem for each choice of~$\TTT$ and~$q$.
The computational complexity of these problem variants are typically
called the ``data complexity'' of the problem.

From a data complexity perspective,
the problems of instance checking and conjunctive query entailment
for the logic \ELI{} are both polynomial-time solvable.
In other words, from a data complexity point of view
these problems are of the same complexity.

\begin{claim}[\citealt{Krisnadhi07}]
\label{claim:eli-ic-p}
Instance checking for \ELI{} is polynomial-time solvable
regarding data complexity.
\end{claim}

\begin{claim}[\citealt{KrisnadhiLutz07}]
\label{claim:eli-cqe-p}
Conjunctive query entailment for \ELI{} is polynomial-time solvable
regarding data complexity.
\end{claim}

\subsection{Parameterized Complexity Results}

We will argue that the computational complexity of the problems
of instance checking and conjunctive query entailment for \ELI{}---%
when only the ABox~$\AAA$ grows in size---%
is of a vastly different nature.
We will do so by using the parameterized complexity methodology.
Concretely, we will take (the size of) the TBox~$\TTT$
and (for the case of conjunctive query entailment) the query~$q$
as parameters, and observe that the parameterized
complexity of these two problems is different.

We begin by observing that the algorithm
witnessing polynomial-time data complexity
for the problem of instance checking for \ELI{}
corresponds to an fpt-algorithm for the problem
when parameterized by the size of the TBox~$\TTT$.

\begin{observation} 
\label{obs:eli-ic-fpt}
Instance checking for \ELI{} is fixed-parameter tractable
when parameterized by~$|\TTT|$.
\end{observation}
\begin{proof}
The algorithm to solve the problem of instance checking for \ELI{}
described by Krisnadhi~\shortcite[Proposition~4.3]{Krisnadhi07}
runs in time~$2^{|\TTT|^{O(1)}} \cdot |\AAA|^{O(1)}$.
\end{proof}

The polynomial-time data complexity algorithm for the problem
of conjunctive query entailment, on the other hand,
does not translate to an fpt-algorithm,
but to an xp-algorithm instead---%
when the parameter is (the sum of) the size of the
TBox~$\TTT$ and the size of the query~$q$.

\begin{observation} 
\label{obs:eli-cqe-xp}
Conjunctive query entailment for \ELI{} is in \XP{}
when parameterized by~$|\TTT|$ and~$|q|$.
\end{observation}
\begin{proof}
The algorithm to solve the problem of conjunctive query entailment for \ELI{}
described by Krisnadhi and Lutz~\shortcite[Theorem~4]{KrisnadhiLutz07}
runs in time~$(|\AAA| + |\TTT|)^{|q|^{O(1)}}$.
\end{proof}

For this parameter, the problem of conjunctive
query entailment for \ELI{} is in fact \W{1}-hard---%
and thus not fixed-parameter tractable,
assuming the widely believed conjecture
that~$\FPT \neq \W{1}$.
This follows immediately from the
\W{1}-hardness of conjunctive query answering
over databases when parameterized by the
size of the query
\cite[Theorem~1]{PapadimitriouYannakakis99}.


\begin{corollary}[{\citealt{PapadimitriouYannakakis99}}]
\label{cor:eli-cqe-w1-hard}
Conjunctive query entailment for \ELI{} is \W{1}-hard
when parameterized by~$|\TTT|$ and~$|q|$.
\end{corollary}

\begin{table}[!b]
  \centering
  
  \begin{small}
  \begin{tabular}{p{2.4cm} @{\ \ \ }|@{\ \ \ } p{2.1cm} @{\ \ \ }|@{\ \ \ } p{2.4cm}} \toprule
      & \textit{instance\newline checking}
      & \textit{conjunctive\newline query entailm.} \\
      \midrule\midrule
      combined\newline complexity &
        \EXPTIME{}-c \newline \phantom{a}\hfill (Prop~\ref{prop:eli-ic-exp}) &
        \EXPTIME{}-c \newline \phantom{a}\hfill (Cor~\ref{cor:eli-cqe-exp}) \\[5pt]
      \midrule
      data complexity &
        in \P{} \hfill (Claim~\ref{claim:eli-ic-p}) &
        in \P{} \hfill (Claim~\ref{claim:eli-cqe-p}) \\[3pt]
      \midrule
      combined\newline complexity with\newline parameter~$|\TTT|+|q|$ &
        in FPT \newline \phantom{a}\hfill (Obs~\ref{obs:eli-ic-fpt}) &
        in XP \hfill (Obs~\ref{obs:eli-cqe-xp}) \newline
        \W{1}-h \hfill (Cor~\ref{cor:eli-cqe-w1-hard}) \\[13pt]
    \bottomrule
  \end{tabular}
  \end{small}
  
  \caption{(Parameterized) complexity results for instance checking and
    conjunctive query entailment for \ELI{}.}
  \label{table:data-complexity}
\end{table}

\subsection{Interpretation of the Results}

The results in this section are summarized
in Table~\ref{table:data-complexity}.
Observation~\ref{obs:eli-ic-fpt}
and Corollary~\ref{cor:eli-cqe-w1-hard} show that parameterized
complexity can give a more accurate view on data complexity
results than classical complexity theory.
From a classical complexity perspective, the data complexity
variants of both problems are polynomial-time solvable,
whereas the parameterized data complexity variants of the
problems differ in complexity.
Both problems are solvable in polynomial
time when only the ABox~$\AAA$ grows in size.
However, for instance checking the order of the polynomial is constant
(Observation~\ref{obs:eli-ic-fpt}), and for conjunctive query entailment
the order of the polynomial grows with the size of the query~$q$
(Corollary~\ref{cor:eli-cqe-w1-hard}).
This is a difference with enormous effects on the practicality of
algorithms solving these problems
(see, e.g.,~\citealt{Downey12}).

\section{Directions for Future Research}

The results in this paper are merely an illustrative exposition
of the type of parameterized complexity results
that are possible for description logic reasoning problems when
using less commonly studied concepts
(e.g., the classes \para{\NP}, \para{\co{\NP}} and \para{\PSPACE}).
We hope that this paper sparks a structured investigation of the
parameterized complexity of different reasoning problems for the
wide range of description logics that have been studied.
For this, it would be interesting to consider a large assortment of
different parameters that could reasonably be expected to have
small values in applications. 
It would also be interesting to investigate to what extent, say,
\para{\NP}-membership results can be used to develop practical
algorithms based on the combination of fpt-time encodings into SAT
and SAT solving algorithms.

\section{Conclusion}

We showed how the complexity study of description logic
problems can benefit from using the framework of parameterized complexity
and all the tools and methods that it offers.
We did so using three case studies.
The first addressed the problem of
concept satisfiability for \ALC{} with respect to nearly acyclic TBoxes.
The second was about  the problem of
concept satisfiability for fragments of \ALC{} that are close
to \ALE{}, \ALU{} and \AL{}, respectively.
The third case study concerned
a parameterized complexity view on the notion of data complexity
for instance checking and conjunctive query entailment for \ELI{}.
Moreover, we sketched some directions for future research,
applying (progressive notions from) parameterized complexity theory
to the study of description logic reasoning problems.

\subsubsection{Acknowledgments.}
This work was supported by the Austrian Science Fund (FWF),
project~J4047.



\DeclareRobustCommand{\DE}[3]{#3}

\bibliographystyle{aaai}

\end{document}